\documentclass{llncs}

\usepackage{amsmath}
\usepackage{amssymb}
\usepackage{stmaryrd} 
\usepackage{wasysym}
\usepackage{mathtools} 

\usepackage{rotating}
\usepackage{pdflscape}

\usepackage{mathpartir}

\usepackage{verbatim}
\usepackage[hidelinks]{hyperref}
\usepackage{color}

\usepackage[noend]{algpseudocode}
\usepackage{algorithm}

\usepackage{cite}

\usepackage{tikz}
\usetikzlibrary{arrows,automata,positioning,petri,topaths}
\usetikzlibrary{decorations.pathmorphing}
\usepackage{pgfplots}
\pgfplotsset{compat=1.14}

\usepackage{tcolorbox}
\tcbset{
		sharp corners,
        colback=white,
        boxrule=.5pt,
        colframe=black,
        fonttitle=\bfseries
       }

\newcommand{\ldot}{\mathpunct{.}}

\newcommand{\dep}{\mathit{dep}}
\newcommand{\var}{\mathit{var}}

\newcommand{\unsat}{\text{UNSAT}}

\newcounter{invariant}
\newenvironment{invariant}[1][]{\refstepcounter{invariant}\par\medskip
   \noindent \textbf{Invariant~\theinvariant. #1} \rmfamily}{\medskip}

\allowdisplaybreaks

\title{Understanding and Extending Incremental Determinization for 2QBF}

\institute{}
\author{Markus N. Rabe${}^1$ \and Leander Tentrup${}^2$ \and Cameron Rasmussen${}^1$ \and Sanjit~A.~Seshia${}^1$}
\institute{
${}^1$University of California, Berkeley\\
${}^2$Saarland University
}

\begin{document}

\maketitle

\begin{abstract}
	Incremental determinization is a recently proposed algorithm for solving quantified Boolean formulas with one quantifier alternation. 
	In this paper, we formalize incremental determinization as a set of inference rules to help understand the design space of similar algorithms. 
	We then present additional inference rules that extend incremental determinization in two ways. 
	The first extension integrates the popular CEGAR principle and the second extension allows us to analyze different cases in isolation. 
	The experimental evaluation demonstrates that the extensions significantly improve the performance. 
\end{abstract}

\section{Introduction}

Solving quantified Boolean formulas (QBFs) is one of the core challenges in automated reasoning and is particularly important for applications in verification and synthesis. 
For example, program synthesis with syntax guidance~\cite{Solar-LezamaRBE/2005/Sketching,AlurBodikJuniwalMartinRaghothamanSeshiaSinghSolarLezamaTorlakUdupa/2015/SyntaxGuidedSynthesis} and the synthesis of reactive controllers from LTL specifications has been encoded in QBF~\cite{FaymonvilleFinkbeinerRabeTentrup/2017/EncodingsOfBoundedSynthesis,Bloemetal/2014/SATbasedSynthesisForSafetySpecs}. 
Many of these problems require only formulas with one quantifier alternation (2QBF), which are the focus of this paper. 

Algorithms for QBF and program synthesis largely rely on the counterex\-am\-ple-guided inductive synthesis principle 
(CEGIS)~\cite{solar-asplos06}, originating in abstraction refinement (CEGAR)~\cite{ClarkeGrumbergJhaLuVeith/2000/CEGAR,JhaSeshia/2017/FormalSynthesisViaInductiveLearning}. 
For example, for program synthesis, CEGIS-style algorithms alternate between generating candidate programs and checking them for counter-examples, which allows us to lift arbitrary verification approaches to synthesis algorithms. 
Unfortunately, this approach often degenerates into a plain guess-and-check loop when counter-examples cannot be generalized effectively. 
This carries over to the simpler setting of 2QBF. 
For example, even for a simple formula such as $\forall x.\exists y.~x=y$, where $x$ and $y$ are 32-bit numbers, most QBF algorithms simply enumerate all $2^{32}$ pairs of assignments. 
In fact, even the modern QBF solvers diverge on this formula when preprocessing is deactivated. 

Recently, Incremental Determinization (ID) has been suggested to overcome this problem~\cite{RabeSeshia/2016/IncrementalDeterminization}. 
ID represents a departure from the CEGIS approach in that it is structured around identifying which variables have unique Skolem functions.
(To prove the truth of a 2QBF $\forall x.\exists y.~\varphi$ we have to find Skolem functions $f$ mapping $x$ to $y$ such that $\varphi[f/y]$ is valid.) 
After assigning Skolem functions to a few of the existential variables, the propagation procedure determines Skolem functions for other variables that are uniquely implied by that assignment. 
When the assignment of Skolem functions turns out to be incorrect, ID analyzes the conflict, derives a conflict clause, and backtracks some of the assignements. 
In other words, ID lifts CDCL to the space of Skolem functions. 

ID can solve the simple example given above and shows good performance on various application benchmarks. 
Yet, the QBF competitions have shown that the relative performance of ID and CEGIS still varies a lot between benchmarks~\cite{conf/sat/Pulina16}. 
A third family of QBF solvers, based on the \emph{expansion} of universal variables~\cite{Biere/2004/ResolveAndExpand,PigorschScholl/2010/AnAIGbasedQBFsolver,CharwatWoltran/2016/DynQBF}, shows yet again different performance characteristics and outperforms both ID and CEGIS on some (few) benchmarks. 
This variety of performance characteristics of different approaches indicates that current QBF solvers could be significantly improved by integrating the different reasoning principles. 

In this paper, we first formalize and generalize ID~\cite{RabeSeshia/2016/IncrementalDeterminization} (Section~\ref{sec:inference-rules}). 
This helps us to disentangle the working principles of the algorithm from implementation-level design choices. 
Thereby our analysis of ID enables a systematic and principled search for better algorithms for quantified reasoning. 
To demonstrate the value and flexibility of the formalization, we present two extensions of ID that integrate CEGIS-style inductive reasoning (Section~\ref{sec:CEGAR}) and expansion (Section~\ref{sec:cases}). 
In the experimental evaluation we demonstrate that both extensions significantly improve the performance compared to plain ID (Section~\ref{sec:eval}).

\vspace{-2mm}
\paragraph{Related work.}
This work is written in the tradition of works such as the Model Evolution Calculus~\cite{BaumgartnerTinelli/2003/TheModelEvolutionCalculus}, AbstractDPLL~\cite{NieuwenhuisOliverasTinelli/AbsractDPLL}, MCSAT~\cite{DejanDeMoura/2013/MCSAT}, and recent calculi for QBF~\cite{conf/synasc/FazekasSB16}, which present search algorithms as inference rules to enable the study and extension of these algorithms. 
ID and the inference rules presented in this paper can be seen as an instantiation of the more general frameworks, such as MCSAT~\cite{DejanDeMoura/2013/MCSAT} or Abstract Conflict Driven Learning~\cite{DSilva/2013/AbstractCDL}. 

Like ID, quantified conflict-driven clause learning (QCDCL) lifts CDCL to QBF~\cite{GiunchigliaNT/2001/QuBE,LonsingBiere/2010/DepQBF}. 
The approaches differ in that QCDCL does not reason about functions, but only about values of variables. 
Fazekas et al.\ have formalized QCDCL as inference rules~\cite{conf/synasc/FazekasSB16}. 
 
2QBF solvers based on CEGAR/CEGIS search for universal assignments and matching existential assignments using two SAT solvers~\cite{Ranjan/2004/AComparativeStudyof2QBFAlgorithms,JanotaSilva/2011/AbstractionBasedAlgorithmFor2QBF,solar-asplos06}. 
There are several generalizations of this approach to QBF with more than one quantifier alternation~\cite{JanotaKMC/2012/QBFWithCEGAR,JanotaMarquesSilva/2015/SolvingQBFByClauseSelection,RabeTentrup/2015/CAQEACertifyingQBFSolver,BloemBH/2016/ijtihad,Tentrup/2017/OnExpansionAndResolution}. 


\section{Preliminaries}
\label{sec:prelim}

\newcommand{\False}{\mathbf{0}}
\newcommand{\True}{\mathbf{1}}

\vspace{-1mm}
Quantified Boolean formulas over a finite set of variables $x\in X$ with domain $\mathbb{B} = \{\False,\True\}$ are generated by the following grammar: 
\begin{equation*}
  \varphi \coloneqq \False \mid \True \mid x \mid \neg \varphi \mid (\varphi) \mid \varphi \lor \varphi \mid \varphi \land \varphi \mid \exists x \ldot \varphi \mid \forall x \ldot \varphi \enspace
\end{equation*}
We consider all other logical operations, including implication, XOR, and equality as syntactic sugar with the usual definitions. 
We abbreviate multiple quantifications $Q x_1.Q x_2.\dots Q x_n.\varphi$ using the same quantifier $Q\in\{\forall,\exists\}$ by the quantification over the set of variables $X=\{x_1,\dots,x_n\}$, denoted as $Q X.\varphi$. 

An \emph{assignment} $\vec x$ to a set of variables $X$ is a function $\vec x : X \rightarrow \mathbb B$ that maps each variable $x \in X$ to either $\True$ or $\False$.
Given a propositional formula $\varphi$ over variables $X$ and an assignment $\vec x'$ to $X'\subseteq X$, we define $\varphi(\vec x')$ to be the formula obtained by replacing the variables $X'$ by their truth value in $\vec x'$.
By $\varphi(\vec x',\vec x'')$ we denote the replacement by multiple assignments for disjoint sets $X',X''\subseteq X$. 

A quantifier $Q\, x \ldot \varphi$ for $Q \in \{\exists,\forall\}$ \emph{binds} the variable $x$ in its subformula $\varphi$ and we assume w.l.o.g. that every variable is bound at most once in any formula.
A \emph{closed} QBF is a formula in which all variables are bound.
We define the dependency set of an existentially quantified variable~$y$ in a formula $\varphi$ as the set $\dep(y)$ of universally quantified variables $x$ such that $\varphi$'s subformula $\exists y \ldot \psi$ is a subformula of $\varphi$'s subformula $\forall x.\psi'$. 
A \emph{Skolem function} $f_y$ maps assignments to $\dep(y)$ to a truth value. 
We define the truth of a QBF $\varphi$ as the existence of Skolem functions $f_Y=\{f_{y_1},\dots,f_{y_n}\}$ for the existentially quantified variables $Y=\{y_1,\dots,y_n\}$, such that $\varphi(\vec x, f_Y(\vec x))$ holds for every $\vec x$, where $f_Y(\vec x)$ is the assignment to $Y$ that the Skolem functions $f_Y$ provide for $\vec x$. 

%
A formula is in \emph{prenex normal form}, if the formula is closed and starts with a sequence of quantifiers followed by a propositional subformula.
A formula $\varphi$ is in the $k$QBF fragment for $k\in\mathbb{N}^+$ if it is closed, in prenex normal form, and has exactly $k-1$ alternations between $\exists$ and $\forall$ quantifiers. 

A \emph{literal} $l$ is either a variable $x \in X$, or its negation $\neg x$. 
Given a set of literals $\{l_1,\dots,l_n\}$, their disjunction $(l_1 \lor \ldots \lor l_n)$ is called a \emph{clause} and their conjunction  $(l_1 \land \ldots \land l_n)$ is called a \emph{cube}.
We use $\overline l$ to denote the literal that is the logical negation of $l$.
We denote the variable of a literal by $\var(l)$ and lift the notion to clauses $\var(l_1\vee\dots\vee l_n) = \{\var(l_1),\dots, \var(l_n)\}$. 

A propositional formula is in conjunctive normal form (CNF), if it is a conjunction of clauses. 
A prenex QBF is in prenex conjunctive normal form (PCNF) if its propositional subformula is in CNF.
Every QBF $\varphi$ can be transformed into an equivalent PCNF with size $O(|\varphi|)$~\cite{Tseitin/1968/OnTheComplexityOfDerivationInPropositionalCalculus}.

\vspace{-2mm}
\paragraph{Resolution} is a well-known proof rule that allows us to merge two clauses as follows.
Given two clauses $C_1\vee v$ and $C_2\vee \neg v$, we call $C_1\otimes_v C_2 = C_1\vee C_2$ their \emph{resolvent} with pivot $v$. 
The resolution rule states that $C_1\vee v$ and $C_2 \vee \neg v$ imply their resolvent. 
Resolution is refutationally complete for propositional Boolean formulas, i.e. for every propositional Boolean formula that is equivalent to false we can derive the empty clause. 

For \emph{quantified} Boolean formulas, however, we need additional proof rules. 
The two most prominent ways to make resolution complete for QBF are to add either \emph{universal reduction} or \emph{universal expansion}, leading to the proof systems Q-resolution~\cite{BuningKarpinskiFlogel/1995/ResolutionforQBF} and $\forall$Exp-Res~\cite{Biere/2004/ResolveAndExpand,JanotaMarquesSilva/2015/ExpansionversusResolution}, respectively. 

\vspace{-2mm}
\paragraph{Universal expansion} eliminates a single universal variable by creating two copies of the subformulas of its quantifier. 
Let $Q_1.\forall x.Q_2.~\varphi$ be a QBF in PCNF, where $Q_1$ and $Q_2$ each are a sequence of quantifiers, and let $Q_2$ quantify over variables $X$. 
Universal expansion yields the \emph{equivalent} formula $Q_1.Q_2.Q_2'.~\varphi[\True/x,X'/X]\wedge\varphi[\False/x]$, where $Q_2'$ is a copy of $Q_2$ but quantifying over a fresh set of variables $X'$ instead of $X$. 
The term $\varphi[\True/x,\ X'/X]$ denotes the $\varphi$ where $x$ is replaced by~$\True$ and the variables $X$ are replaced by their counterparts in $X'$. 

\vspace{-2mm}
\paragraph{Universal reduction} allows us to drop universal variables from clauses when none of the existential variables in that clause may depend on them. 
Let $C$ a clause of a QBF and let $l$ be a literal of a universally quantified variable in $C$. 
Let us further assume that $\overline l$ does not occur in $C$. 
If all existential variables $v$ in $C$ we have $\var(l)\notin\dep(v)$, universal reduction allows us to remove $l$ from $C$. 
The resulting formula is equivalent to the original formula. 


\newcommand{\add}{\mathit{add}}

\vspace{-2mm}
\paragraph{Stack.}
For convenience, we use a stack data structure to describe the algorithm. 
Formally, a stack is a finite sequence. 
Given a stack $S$, we use $S(i)$ to denote the $i$-th element of the stack, starting with index 0, and we use $S.S'$ to denote concatenation. 
We use $S[0,i]$ to denote the prefix up to element $i$ of $S$.
All stacks we consider are stacks of sets. 
In a slight abuse of notation, we also use stacks as the union of their elements when it is clear from the context. 
We also introduce an operation specific to stacks of sets $S$: We define $\add(S,i,x)$ to be the stack that results from extending the set on level $i$ by element $x$.

\subsection{Unique Skolem Functions}
\newcommand{\UC}{\mathcal{U}}

Incremental determinization builds on the notion of unique Skolem functions. 
%
Let $\forall X. \exists Y.\ \varphi$ be a 2QBF in PCNF and let $\chi$ be a formula over $X$ characterizing the \emph{domain} of the Skolem functions we are currently interested in. 
We say that a variable $v\in Y$ has a \emph{unique Skolem function} for domain $\chi$, if for each assignment $\vec x$ with $\chi(\vec x)$ there is a \emph{unique} assignment $\vec v$ to $v$ such that $\varphi(\vec x,\vec v)$ is satisfiable. 
In particular, a unique Skolem function is a Skolem function:

\begin{lemma}
\label{lemma:uniqueimpliesexists}
	If all existential variables have a unique Skolem function for the full domain $\chi=\True$, the formula is true. 
\end{lemma}

\vspace{-1mm}
The semantic characterization of unique Skolem functions above does not help us with the computation of Skolem functions directly. 
We now introduce a local approximation of unique Skolem functions and show how it can be used as a propagation procedure. 

We consider a set of variables $D\subseteq X\cup Y$ with $D\supseteq X$ and focus on the subset $\varphi|_D$ of clauses that only contain variables in $D$. 
We further assume that the existential variables in $D$ already have unique Skolem functions for $\chi$ in the formula $\varphi|_D$. 
We now define how to extend $D$ by an existential variable $v\notin D$. 
To define a Skolem function for $v$ we only consider the clauses with \emph{unique consequence} $v$, denoted $\UC_v$, that contain a literal of $v$ and otherwise only literals of variables in $D$. 
(Note that $\varphi|_D\cup\UC_v=\varphi|_{D\cup\{v\}}$). 
We define that variable $v$ has a \emph{unique Skolem function relative to $D$} for $\chi$, if for all assignments to $D$ satisfying $\chi$ and $\varphi$ there is a unique assignment to $v$ satisfying $\UC_v$. 


\newcommand{\deterministic}{\mathsf{deterministic}}
\newcommand{\conflict}{\mathsf{conflict}}
\newcommand{\unconflicted}{\mathsf{unconflicted}}

In order to determine unique Skolem functions relative to a set $D$ in practice, we split the definition into the two statements $\deterministic$ and $\unconflicted$. 
Each statement can be checked by a SAT solver and together they imply that variable $v$ has a unique Skolem function relative to $D$. 

Given a clause $C$ with unique consequence $v$, let us call $\neg (C\setminus\{v,\neg v\})$ the \emph{antecedent} of $C$. 
Further, let $\mathcal A_{l} = \bigvee_{C\in\UC_v,l\in C} \neg (C\setminus\{v,\neg v\})$ be the disjunction of antecedents for the unique consequences containing the literal $l$ of $v$. 
It is clear that whenever $\mathcal A_{v}$ is satisfied, $v$ needs to be true, and whenever $\mathcal A_{\neg v}$ is satisfied, $v$ need to be false. 
 We define:
\[
\begin{array}{lllll}
    \deterministic(v,\varphi,\chi,D) &:=~&\forall D.~\varphi|_D\wedge \chi &\Rightarrow &\mathcal A_v \vee  \mathcal A_{\neg v} \quad \\
    \unconflicted(v,\varphi,\chi,D) &:=& \forall D.~ \varphi|_D \wedge\chi ~&\Rightarrow~ \neg(&\mathcal A_v \wedge \mathcal A_{\neg v}\;)
\end{array}
\]
$\deterministic$ states that $v$ needs to be assigned either true or false for every assignment to $D$ in the domain $\chi$ that is consistent with the existing Skolem function definitions $\varphi|_D$. 
Accordingly, $\unconflicted$ states that $v$ does not have to be true and false at the same time (which would indicate a conflict) for any such assignment. 
%
%
%
%
Unique Skolem functions relative to a set $D$ approximate unique Skolem functions as follows:
\begin{lemma}
\label{lemma:uniqueSkolemInduction}
	Let the existential variables in $D$ have unique Skolem functions for domain $\chi$ and let $v\in Y$ have a unique Skolem function relative to $D$ for domain~$\chi$. 
	Then $v$ has a unique Skolem function for domain $\chi$.
\end{lemma}


\section{Inference Rules for Incremental Determinization}
\label{sec:inference-rules}

In this section, we develop a nondeterministic algorithm that formalizes and generalizes ID. 
We describe the algorithm in terms of inference rules that specify how the state of the algorithm can be developed. 
The state of the algorithm consists of the following elements:
\newcommand{\Ready}{\mathsf{Ready}}
\newcommand{\Conflict}{\mathsf{Conflict}}
\newcommand{\SAT}{\mathsf{SAT}}
\newcommand{\UNSAT}{\mathsf{UNSAT}}
\newcommand{\Sat}{\mathit{Sat}}
\newcommand{\dlvl}{\mathit{dlvl}}
\begin{itemize}
	\item The solver status $S\in\{\Ready,\Conflict(L,\vec x),\SAT,\UNSAT\}$. The conflict status has two parameters: a clause $L$ that is used to compute the learnt clause and the assignment $\vec x$ to the universals witnessing the conflict.
	\item A stack $C$ of sets of clauses. $C(0)$ contains the original and the learnt clauses. $C(i)$ for $i>0$ contain temporary clauses introduced by decisions. 
	\item A stack $D$ of sets of variables. The union of all levels in the stack represent the set of variables that currently have unique Skolem functions and the clauses in $C|_D$ represent these Skolem functions. $D(0)$ contains the universals and the existentials whose Skolem functions do not depend on decisions. 
	\item A formula $\chi$ over $D(0)$ characterizing the set of assignments to the universals for which we still need to find a Skolem function. 
	\item A formula $\alpha$ over variables $D(0)$ representing a \emph{temporary} restriction on the domain of the Skolem functions.
\end{itemize}

We assume that we are given a 2QBF in PCNF $\forall X. \exists Y.~\varphi$ and that all clauses in $\varphi$ contain an existential variable.
(If $\varphi$ contains a non-tautological clause without existential variables, the formula is trivially false by universal reduction.)
We define $(\Ready,\varphi,X, \True, \True)$ to be the initial state of the algorithm. 
That is, the clause stack $C$ initially has height~1 and contains the clauses of the formula $\varphi$.
We initialize $D$ as the stack of height~1 containing the universals.

Before we dive into the inference rules, we want to point out that some of the rules in this calculus are not computable in polynomial time. 
The judgements $\deterministic$ and $\unconflicted$ require us to solve a SAT problem and are, in general, NP-complete. 
This is still easier than the 2QBF problem itself (unless NP includes $\Pi_2^P$) and in practice they can be discharged quickly by SAT solvers.

\subsection{True QBF} 

We continue with describing the basic version of ID, consisting of the rules in Fig.~\ref{fig:true} and Fig.~\ref{fig:false}, and first focus on the rules in Fig.~\ref{fig:true}, which suffice to prove true 2QBFs. 
\textsc{Propagate} allows us to add a variable to $D$, if it has a unique Skolem function relative to $D$.
(The notation $\add(D,|D|-1,v)$ means that we add $v$ to the last level of the stack.) 
The judgements $\deterministic$ and $\unconflicted$ involve the current set of clauses $C$ (i.e. the union of all sets of clauses in the sequence $C$).
These checks are restricted to the domain $\chi\wedge\alpha$. Both $\chi$ and $\alpha$ are true throughout this section; we discuss their use in Section~\ref{sec:CEGAR} and Section~\ref{sec:cases}. 

\begin{figure}[t]
\vspace{-2mm}
\begin{center}
\begin{tcolorbox}
\begin{minipage}{12.1cm}
$
\phantom{Propagate}~
\inferrule*[Left=Propagate]{(\Ready,C,D,\chi,\alpha) \\ v\notin D \\\\ \deterministic(v,C,\chi\wedge\alpha,D) \\\unconflicted(v,C,\chi\wedge\alpha,D)}
          {(\Ready,C,\add(D,|D|-1,v),\chi,\alpha)}
$
\end{minipage}

\vspace{5mm}
\begin{minipage}{11.5cm}
$
\phantom{Decide}~
\inferrule*[Left=Decide]{(\Ready,C,D,\chi,\alpha) \\ v\notin D \\ \text{all }c\in\delta\text{ have unique consequence } v}
          {(\Ready, C.\delta, D.\emptyset, \chi, \alpha)}
$
\end{minipage}

\vspace{5mm}
\begin{minipage}{12.1cm}
$
\phantom{Sat}~
\inferrule*[Left=Sat]{(\Ready,C,D,\chi,\True) \\ D=X\cup Y}
          {(\SAT,C,D,\chi,\True)}
$
\end{minipage}
\end{tcolorbox}
\end{center}
\vspace{-5mm}
\caption{Inference rules needed to prove true QBF}
\label{fig:true}
\end{figure}

\begin{invariant}
\label{inv:Dconsistency}
	All existential variables in $D$ have a unique Skolem function for the domain $\chi\wedge\alpha$ in the formula $\forall X. \exists Y.~ C|_D$, where $C|_D$ are the clauses in $C$ that contain only variables in $D$. 
\end{invariant}

If \textsc{Propagate} identifies all variables to have unique Skolem functions relative to the growing set $D$, we know that they also have unique Skolem functions (Lemma~\ref{lemma:uniqueSkolemInduction}).
We can then apply \textsc{Sat} to reach the $\SAT$ state, representing that the formula has been proven true (Lemma~\ref{lemma:uniqueimpliesexists}).

\begin{lemma}
\label{prop:soundnessFalse}
ID cannot reach the SAT state for false QBF.
\end{lemma}
\begin{proof}
  Let us assume we reached the $\SAT$ state for a false 2QBF and prove the statement by way of contradiction. 
  The $\SAT$ state can only be reached by the rule \textsc{Sat} and requires $D = X \cup Y$. 
  By Invariant~\ref{inv:Dconsistency} all variables have a Skolem function in $\forall X.\exists Y.~ C$. 
  Since $C$ includes $\varphi$, this Skolem function does not violate any clause in $\varphi$, which means it is indeed a proof.
  \qed
\end{proof}

When \textsc{Propagate} is unable to determine the existence of a unique Skolem function (i.e. for variables where the judgement $\deterministic$ does not hold) we can use the rule \textsc{Decide} to introduce additional clauses such that $\deterministic$ holds and propagation can continue. 
Note that additional clauses make it easier to satisfy $\deterministic$ and adding the clause $v$ (i.e. a unit clause) even \emph{ensures} that $\deterministic$ holds for $v$. 

Assuming we consider a true 2QBF, we can pick a Skolem function $f_{y}$ for each existential variable $y$ and encode it using \textsc{Decide}. 
We can simply consider the truth table of $f_y$ in terms of the universal variables and define $\delta$ to be the set of clauses $\{\neg\vec x \vee v\mid f_y(\vec x)\}\cup\{\neg\vec x \vee \neg v\mid \neg f_y(\vec x)\}$. (Here we interpret the assignment $\vec x$ as a conjunction of literals.)
These clauses have unique consequence $v$ and they guarantee that $v$ is deterministic.
Further, they guarantee that $v$ is $\unconflicted$, as otherwise $f_y$ would not be a Skolem function, so we can apply \textsc{Propagate} to add $v$ to $D$.
Repeating this process for every variable let us reach the point where $Y\subseteq D$ and we  can apply \textsc{Sat} to reach the $\SAT$ state. 

\begin{lemma}
\label{prop:completenessTrue}
ID can reach the SAT state for true QBF.
\end{lemma}

Note that proving the truth of a QBF in this way requires guessing correct Skolem functions for all existentials. 
In Subsection~\ref{ssec:termination} we discuss how termination is guaranteed with a simpler type of decisions. 

\begin{figure}[t]
\vspace{-2mm}
\begin{center}
\begin{tcolorbox}

\begin{minipage}{12.1cm}
$	
\phantom{Conflict}~
\inferrule*[Left=Conflict]{(\Ready,C,D,\chi,\alpha) \\ \vec x\text{ refutes }\unconflicted(v,C,\chi\wedge\alpha,D)}
          {(\Conflict(\{v,\neg v\},\vec x),C,D,\chi,\alpha)}
$
\end{minipage}

\vspace{5mm}
\begin{minipage}{12.1cm}
$
\phantom{Analyze}~
\inferrule*[Left=Analyze]{(\Conflict(L,\vec x),C,D,\chi,\alpha) \\ c\in C(0) \\ l\in L \\  \overline l\in c }
          {(\Conflict(L \otimes_{\var(l)} c,\vec x),C,D,\chi,\alpha)}
$
\end{minipage}

\vspace{5mm}
\begin{minipage}{12.1cm}
$
\phantom{Learn}~
\inferrule*[Left=Learn]{(\Conflict(L,\vec x),C,D,\chi,\alpha) \\ \var(L)\not\subseteq D}
          {(\Ready,\add(C,0,L),D,\chi,\alpha)}
$
\end{minipage}

\vspace{5mm}
\begin{minipage}{12.1cm}
$
\phantom{Unsat}~
\inferrule*[Left=Unsat]{(\Conflict(L,\vec x),C,D, \chi, \alpha) \\ \var(L)\subseteq D(0) \\ \vec x \not\models L}
          {(\UNSAT,C,D,\chi,\alpha)}
$
\end{minipage}

\vspace{5mm}
\begin{minipage}{12.1cm}
$
\phantom{Backtrack}~~
\inferrule*[Left=Backtrack]{(S,C,D,\chi,\alpha) \\ 0<\dlvl \leq |C|}
          {(S,C[0,{\dlvl}],D[0,{\dlvl}],\chi,\alpha)}
$
\end{minipage}
\end{tcolorbox}
\end{center}
\vspace{-6mm}
\caption{Additional inference rules needed to disprove false QBF}
\label{fig:false}
\end{figure}

\subsection{False QBF}
\label{ssec:unsat}

To disprove false 2QBFs, i.e. formulas that do not have a Skolem function, we need the rules in Fig.~\ref{fig:false} in addition to \textsc{Propagate} and \textsc{Decide} from Fig.~\ref{fig:true}. 
The $\conflict$ state can only be reached via the rule \textsc{Conflict}, which requires that a variable $v$ is conflicted, i.e. $\unconflicted$ fails. 
The \textsc{Conflict} rule stores the assignment $\vec x$ to $D$ that proves the conflict and it creates the nucleus of the learnt clause $\{v,\neg v\}$. 
Via \textsc{Analyze} we can then resolve that nucleus with clauses in $C(0)$, which consists of the original clauses and the clauses learnt so far. 
We are allowed to add the learnt clause back to $C(0)$ by applying \textsc{Learn}. 

\begin{invariant}
\label{inv:phieqc0}
	$C(0)$ is equivalent to $\varphi$.
\end{invariant}

Note that $C(0)$ and $\varphi$ are propositional formulas over $X\cup Y$. 
Their equivalence means that they have the same set of satisfying assignments. 
We prove Invariant~\ref{inv:phieqc0} together with the next invariant. 

\begin{invariant}
\label{inv:Limpliedbyphi}
	Clause $L$ in conflict state $\Conflict(L,\vec x)$ is implied by $\varphi$.
\end{invariant}

\vspace{-2mm}
\begin{proof}
	$C(0)$ contains $\varphi$ initially and is only ever changed by adding clauses through the \textsc{Learn} rule, so $C(0)\Rightarrow\varphi$ holds throughout the computation.

	We prove the other direction of Invariant~\ref{inv:phieqc0} and Invariant~\ref{inv:Limpliedbyphi} by mutual induction.
	Initially, $C(0)$ consists exactly of the clauses $\varphi$, satisfying Invariant~\ref{inv:phieqc0}.
	The nucleus of the learnt clause $v\vee\neg v$ is trivially true, so it is implied by any formula, which gives us the base case of Invariant~\ref{inv:Limpliedbyphi}. 
	\textsc{Analyze} is the only rule modifying $L$, and hence soundness of resolution together with Invariant~\ref{inv:phieqc0} already gives us the induction step for Invariant~\ref{inv:Limpliedbyphi}~\cite{Robinson/1965/Resolution}.
	Since \textsc{Learn} is the only rule changing $C(0)$, Invariant~\ref{inv:Limpliedbyphi} implies the induction step of Invariant~\ref{inv:phieqc0}. \qed
\end{proof}

When adding the learnt clause to $C(0)$ we have to make sure that Invariant~\ref{inv:Dconsistency} is preserved. 
\textsc{Learn} hence requires that we have backtracked far enough with \textsc{Backtrack}, such that at least one of the variables in $L$ is not in $D$ anymore. 
In this way, $L$ may become part of future Skolem function definitions, but will first have to be checked for causing conflicts by \textsc{Propagate}. 
	
If all variables in $L$ are in $D(0)$ and the assignment $\vec x$ from the conflict violates $L$, we can conclude the formula to be false using \textsc{Unsat}. 
The soundness of this step follows from the fact that $\vec x$ includes an assignment satisfying $C(0)|_{D(0)}$ (i.e. the clauses defining the Skolem functions for $D(0)$), Invariant~\ref{inv:Dconsistency} and Invariant~\ref{inv:Limpliedbyphi}.

\begin{lemma}
\label{prop:soundnessTrue}
ID cannot reach the UNSAT state for true QBF.
\end{lemma}

We will now show that we can disprove any false QBF. 
The main difficulty in this proof is to show that from any $\Ready$ state we can learn a \emph{new} clause, i.e. a clause that is semantically different to any clause in $C(0)$, and then return to the $\Ready$ state. 
Since there are only finitely many semantically different clauses over variables $X\cup Y$, and we cannot terminate in any other way (Lemma~\ref{prop:soundnessTrue}), we eventually have to find a clause $L$ with $\var(L)\subseteq D(0)$, which enables us to go to the $\UNSAT$ state. 

From the $\Ready$ state, we can always add more variables to $D$ with \textsc{Decide} and \textsc{Propagate}, until we reach a conflict.
(Otherwise we would reach a state where $D=Y$ we were able to prove $\SAT$, contradicting Lemma~\ref{prop:soundnessTrue}.)
We only enter a $\Conflict$ state for a variable $v$, if there are two clauses $(c_1\vee v)$ and $(c_2\vee\neg v)$ with unique consequence $v$ such that $\vec x\models \neg c_1 \land \neg c_2$ (see definition of $\unconflicted$). 

In order to apply \textsc{Analyze}, we need to make sure that $(c_1\vee v)$ and $(c_2\vee\neg v)$ are in $C(0)$.
We can guarantee this by restricting \textsc{Decide} as follows:
We say a decision for a variable $v'$ is \emph{consistent with the unique consequences} in state $(\Ready,C,D,\chi,\alpha)$, if $\unconflicted(v,C.\delta,\chi\wedge\alpha,D)$. 
We can construct such a decision easily by applying \textsc{Decide} only on variables that are not conflicted already (i.e. $\unconflicted(v,C,\chi\wedge\alpha,D)$) and by defining $\delta$ to be the CNF representation of $\neg\mathcal A_{v}\Rightarrow \neg v$ (i.e. require $v$ to be false, unless a unique consequence containing literal $v$ applies).
It is clear that for this $\delta$ no new conflict for $v$ is introduced and hence $\unconflicted(v,C.\delta,\chi\wedge\alpha,D)$. 

Assuming that all decisions are taken consistent with the unique consequences, we know that when we encounter a conflict for variable $v$, we did not apply \textsc{Decide} for $v$, and hence the clauses $(c_1\vee v)$ and $(c_2\vee\neg v)$ causing the conflict must be in $C(0)$.
We can hence apply \textsc{Analyze} twice with clauses $(c_1\vee v)$ and $(c_2\vee\neg v)$ and obtain the learnt clause $L=c_1\vee c_2$.
Since $\vec x\models \neg c_1 \land \neg c_2$, the learnt clause is violated by $\vec x$. 
As $\vec x$ refutes $\unconflicted(v,C,\chi\land\alpha,D)$ by construction, it must satisfy the clauses $C|_D$ and learnt clause $L$ hence cannot be in $C|_D$. 
Further, $L$ only contains variables that are in $D$, as $(c_1\vee v)$ and $(c_2\vee\neg v)$ were clauses with unique consequence $v$. 
So, $L$ would have been in $C|_D$, if it existed in $C$ already, and hence $L$ is new. 
We can either add the new clause to $C(0)$ after backtracking, or we can conclude $\UNSAT$.

\begin{lemma}
\label{prop:completenessFalse}
ID can reach the $\UNSAT$ state for false QBF.	
\end{lemma}

The clause learning process considered here only applies one actual resolution step per conflict ($L_1\otimes_v L_2$). 
In practice, we probably want to apply multiple resolution steps before applying \textsc{Learn}. 
It is possible to use the conflicting assignment $\vec x$ to (implicitly) construct an implication graph and mimic the clause learning of SAT solvers~\cite{RabeSeshia/2016/IncrementalDeterminization,MarquesSilvaSakallah/1997/GRASPaNewSearchAlgorithmForSatisfiability}.

\subsection{Example}
\label{ssec:example}

We now discuss the application of the inference rules along the following formula:
\[
\begin{array}{rl}
    \forall x_1,x_2.~\exists y_1,\dots,y_4. 
    & (x_1\lor \neg y_1) ~\wedge~ (x_2\lor\neg y_1) \!~\wedge~ (\neg x_1\lor \neg x_2\lor y_1)\wedge \quad\hfill (1)\\
    & (\neg x_2 \lor y_2) ~\wedge~ (\neg y_1\lor y_2) ~\wedge~ (x_2\lor y_1 \lor \neg y_2)\wedge\hfill (2)\\
    & (y_1 \lor \neg y_3) ~\wedge~ (y_2\lor \neg y_3) ~\wedge\hfill (3)\\
    & (\neg y_1 \vee y_4) ~\wedge~ (\neg y_3\vee \neg y_4) \hfill (4)
\end{array}
\]

Initially, the state of the algorithm is the tuple $(\Ready, \varphi, X, \True, \True)$. 
The rule \textsc{Propagate} can be applied to $y_1$ in the initial state, as we are in the $\Ready$ state, $y_1\notin X$, and because $y_1$ satisfies the checks $\deterministic$ and $\unconflicted$: The antecedents of $y_1$ are $\mathcal A_{y_1}=x_1\wedge x_2$ and $\mathcal A_{\neg y_1}=\neg x_1\vee \neg x_2$ (see clauses in line $(1)$). 
It is easy to check that both $\mathcal A_{y_1} \vee \mathcal A_{\neg y_1}$ nor $\neg(\mathcal A_{y_1}\wedge\mathcal A_{\neg y_1})$ hold for all assignments to $x_1$ and $x_2$. 
The state resulting from the application of \textsc{Propagate} is $(\Ready, \varphi, X\cup\{y_1\}, \True, \True)$.
(Alternatively, we could apply \textsc{Decide} in the initial state, but deriving unique Skolem functions is generally preferable.)

While \textsc{Propagate} was not applicable to $y_2$ before, it now is, as the increased set $D$ made $y_2$ $\deterministic$ (see clauses in line $(2)$).
We can thus derive the state $(\Ready, \varphi, X\cup\{y_1,y_2\}, \True, \True)$. 

Now, we ran out of variables to propagate and the only applicable rule is \textsc{Decide}. 
We arbitrarily choose $y_3$ as our decision variable and arbitrarily introduce a single clause $\delta = \{(\neg y_1\vee \neg y_2\vee y_3)\}$, arriving in the state $(\Ready, \varphi.\delta, X\cup\{y_1,y_2\}, \True, \True)$.
In this step we can immediately apply \textsc{Propagate} (consider $\delta$ and the clauses in line $(3)$) to add the decision variable to the set $D$ and arrive at $(\Ready, \varphi.\delta, X\cup\{y_1,y_2,y_3\}, \True, \True)$. 

We can now apply \textsc{Backtrack} to undo the last decision, but this would not be productive.
Instead identify $y_4$ to be conflicted and we enter a conflict state with \textsc{Conflict}:~ 
$(\Conflict(\{y_4,\neg y_4\}, x_1\wedge x_2), \varphi.\delta, X\cup\{y_1,y_2,y_3\}, \True, \True)$.
To resolve the conflict we apply \textsc{Analyze} twice - once with each of the clauses in line $(4)$ - bringing us into state $(\Conflict(\{\neg y_1,\neg y_3\}, x_1\wedge x_2), \varphi.\delta, X\cup\{y_1,y_2,y_3\}, \True, \True)$. 
We can backtrack one level such that $D=X\cup\{y_1,y_2\}$ and then apply \textsc{Learn} to enter state $(\Ready, \varphi\cup\{(\neg y_1\vee \neg y_3)\}, X\cup\{y_1,y_2\}, \True, \True)$. 

The rest is simple: we apply \textsc{Propagate} on $y_3$ and take a decision for $y_4$. 
As no other variable can depend on $y_4$ we can take an arbitrary decision for $y_4$ that makes $y_4$ deterministic, as long as this does not make $y_4$ conflicted. 
Finally, we can propagate $y_4$ and then apply $\SAT$ to conclude that we have found Skolem functions for all existential variables.

\subsection{Termination}
\label{ssec:termination}
So far, we have described sound and nondeterministic algorithms that allow us to prove or disprove any 2QBF. 
We can easily turn the algorithm in the proof of Lemma~\ref{prop:completenessFalse} into a \emph{deterministic} algorithm that terminates for both true and false QBF by introducing an arbitrary ordering of variables and assignments: 
Whenever there is nondeterminism in the application of one of the rules as described in Lemma~\ref{prop:completenessFalse}, pick the smallest variable for which one of the rules is applicable. 
When multiple rules are applicable for that variable, pick them in the order they appear in the figures. 
When the inference rule allows multiple assignments, pick the smallest. 
In particular, this guarantees that the existential variables are added to $D$ in the arbitrarily picked order, as for any existential not in $D$ we can either apply \textsc{Propagate}, \textsc{Decide}, or \textsc{Conflict}.

Restricting \textsc{Decide} to decisions that are consistent with the unique consequences may be unintuitive for true QBF, where we try to find a Skolem function. 
However, whenever we make the 2QBF false by introducing clauses with \textsc{Decide}, we will eventually go to a conflict state and learn a new clause.
Deriving the learnt clause for conflicted variable $v$ from two clauses with unique consequence $v$ (as described for Lemma~\ref{prop:completenessFalse}) means that we push the constraints towards \emph{smaller} variables in the variable ordering. 
The learnt clause will thus improve the Skolem function for a smaller variable or cause another conflict for a smaller variable. 
In the extreme case, we will eventually learn clauses that look like function table entries, as used in Lemma~\ref{prop:completenessTrue}, i.e. clauses containing exactly one existential variable.
At some point, even with our restriction for \textsc{Decide}, we cannot make a ``wrong'' decision: The cases for which a variable does not have a clause with unique consequence are either irrelevant for the satisfaction of the 2QBF or our restricted decisions happen to make the right assignment. 

In cases where no static ordering of variables is used - as it will be the case in any practical approach - the termination for true QBF is less obvious but follows the same argument: 
Given enough learnt clauses, the relationships between the variables are dense enough such that even naive decisions suffice.

\subsection{Pure literals}
\label{ssec:pure}
The original paper on ID introduces the notion of \emph{pure literals} for QBF that allows us to propagate a variable $v$ even if it is not deterministic, if for a literal $l$ of $v$, all clauses $c$ that $l$ occurs in are either satisfied or $l$ is the unique consequence of $c$. 
The formalization presented in this section allows us to conclude that pure literals are a special case of \textsc{Decide}: 
We can introduce clauses defining $v$ to be of polarity $\overline{l}$ whenever all clauses containing $l$ are satisfied by another literal. 

That is, we can precisely characterize the minimal set of cases in which $v$ has to be of polarity $l$ and the decision is guaranteed to never introduce unnecessary conflicts. 
The same definition cannot be made when $l$ occurs in clauses where it is not a unique consequence, as then the clause contains another variable that is not deterministic yet. 

\subsection{Relation of ID and CDCL}
\label{ssec:sat}

\begin{figure}[t]
\vspace{-3mm}
\begin{center}
\begin{tabular}{l|l|l}
& SAT\quad$\exists Y.~~ \varphi$ & 2QBF\quad$\forall X.~\exists Y.~~ \varphi$\\[2pt]\hline
State & Partial assignment of values~ & Partial assignment of functions \\
Propagation & unit propagation & unique Skolem function w.r.t. $D$ \\
Decision & unit clause & clause with unique consequence \\
Conflict & unit clauses $y$ and $\neg y$ & $\exists X$ that implies $y$ and $\neg y$ \\
Learning & clause & clause
\end{tabular}
\end{center}
\vspace{-3mm}
\caption{Concepts in ID and their counterparts in CDCL}
\label{fig:comparisontoCDCL}
\end{figure}

There are some obvious similarities between ID and conflict-driven clause learning (CDCL) for SAT. 
Both algorithms modify their partial assignments by propagation, decisions, clause learning, and backtracking.
The main difference between the algorithms is that, while CDCL solvers maintain a partial assignment of Boolean values to variables, ID maintains a partial assignment of functions to variables (which is represented by the clauses $C|_D$). 
We summarized our observations in Fig.~\ref{fig:comparisontoCDCL}.

\begin{figure}[t]
\vspace{-2mm}
\begin{center}
\begin{tcolorbox}
\begin{minipage}{10.5cm}
$	
\phantom{InductiveRefinement}~~
\inferrule*[Left=InductiveRefinement]{(\Conflict(L,\vec x),C,D,\chi,\alpha) \\ \varphi(\vec x|_X, \vec y)}
          {(\Conflict(L,\vec x),C,D,\chi\wedge \neg \varphi(\vec y),\alpha)} 
$
\end{minipage}

\vspace{5mm}
\begin{minipage}{10.7cm}
$	
\phantom{Failed}~~
\inferrule*[Left=Failed]{(\Conflict(L,\vec x),C,D,\chi,\alpha) \\ \varphi(\vec x|_X) \text{ is unsatisfiable}}
          {(\unsat,C,D,\chi,\alpha)}
$
\end{minipage}
\end{tcolorbox}
\end{center}
\vspace{-6mm}
\caption{Inference rules adding inductive reasoning to ID}
\label{fig:CEGAR}
\end{figure}

\section{Inductive Reasoning}
\label{sec:CEGAR}

The CEGIS approach to solving a 2QBF $\forall X \ldot \exists Y \ldot \varphi$ is to iterate over $X$ assignments $\vec x$ and check if there is an assignment $\vec y$ such that $\varphi(\vec x, \vec y)$ is valid. 
Upon every successful iteration we exclude all assignments to $X$ for which $\vec y$ is a matching assignment. 
If the space of $X$ assignments is exhausted we conclude the formula is true, and if we find an assignment to $X$ for which there is no matching $Y$ assignment, the formula is false~\cite{JanotaSilva/2011/AbstractionBasedAlgorithmFor2QBF}. 

While this approach shows poor performance on some problems, as discussed in the introduction, it is widely popular and has been successfully applied in many cases. 
In this section we present a way how it can be integrated in ID in an elegant way. 
The simplicity of the CEGIS approach carries over to our extension of ID - we only need the two additional inference rules in Fig.~\ref{fig:CEGAR}.

We exploit the fact that ID already generates assignments $\vec x$ to $X$ in its conflict check. 
Whenever ID is in a conflict state, the rules in Fig.~\ref{fig:CEGAR} allow us to check if there is an assignment $\vec y$ to $Y$ that together with $\vec|_X$, which is the part of $\vec x$ defining variables in $X$, satisfies $\varphi$. 
If there is such an assignment $\vec y$, we can let the Skolem functions output $\vec y$ for the input $\vec x$. 
But the output $\vec y$ may work for other assignments to $X$, too. 
The set of all assignments to $X$ for which $\vec y$ works as an output, is easily characterized by $\varphi(\vec y)$. 
\footnote{We can actually exploit the Skolem functions that do not depend on decisions and exclude $C(0)(\vec y_{\overline{D(0)}})$ from $\chi$ instead, i.e. the set of assignments to $D(0)$ to which the part of $\vec y$ that is not in $D(0)$ is a solution.}
\textsc{InductiveRefinement} allows us to exclude the assignments $\varphi(0)$ from $\chi$, which represents the domain (i.e. assignments to $X$) for which we still need to find a Skolem function. 

This gives rise to a new invariant, stating that $\neg\chi$ only includes assignments to $X$ for which we know that there is an assignment to $Y$ satisfying $\varphi$. 
With this invariant it is clear that Lemma~\ref{prop:soundnessFalse} also holds for arbitrary $\chi$.

\begin{invariant}
\label{inv:psi}
	$\forall X.\exists Y.~\neg\chi \Rightarrow \varphi$
\end{invariant}


It is easy to check that \textsc{Propagate} preserves Invariant~\ref{inv:Dconsistency} also if $\chi$ and $\alpha$ are not $\True$. 
Invariant~\ref{inv:phieqc0} and Invariant~\ref{inv:Limpliedbyphi} are unaffected by the rules in this section. 
To make sure that Lemma~\ref{prop:soundnessTrue} is preserved as well, we thus only have to inspect \textsc{Failed}, which is trivially sound. 

\vspace{-2mm}
\paragraph{A portfolio approach?}
In principle, we could generate assignments $\vec x$ independently from the conflict check of ID. 
The result would be a portfolio approach that simply executes ID and CEGIS in parallel and takes the result from whichever method terminates first. 
The idea behind our extension is that conflict assignments are more selective and may thus increase the probability that we hit a refuting assignment to $X$. 
Also ID may profit from excluding groups of assignments for which frequently cause conflicts. 
We revisit this question in Section~\ref{sec:eval}.

\vspace{-2mm}
\paragraph{Example.} We extend the example from Subsection~\ref{ssec:example} from the point where we entered the conflict state $(\Conflict(\{y_4,\neg y_4\}, x_1\wedge x_2), \varphi.\delta, X\cup\{y_1,y_2,y_3\}, \True, \True)$. 
We can apply \textsc{InductiveRefinement}, checking that there is indeed a solution to $\varphi$ for the assignment $x_1, x_2$ to the universals (e.g. $y_1,y_2,\neg y_3, y_4$).
Instead of doing the standard conflict analysis as in our previous example, we can apply \textsc{Learn} to add the (useless) clause $y_4\vee\neg y_4$ to $C(0)$ without any backtracking.
That is, we effectively ignore the conflict and go to state $(\Ready, \varphi\cup\{(y_4\vee\neg y_4)\}.\delta, X\cup\{y_1,y_2,y_3\}, \neg x_1 \vee \neg x_2, \True)$. 

There is no assignment to $X$ that provokes a conflict for $y_4$, other than the one we excluded through \textsc{InductiveRefinement}. 
We can thus take an arbitrary decision for $y_4$ that is consistent with the unique consequences (see Subsection~\ref{ssec:unsat}), \textsc{Propagate} $y_4$, and then conclude the formula to be true.

%
%

\section{Expansion}
\label{sec:cases}

Universal expansion (defined in Section~\ref{sec:prelim}) is another fundamental proof rule that deals with universal variables. 
It has been used in early QBF solvers~\cite{Biere/2004/ResolveAndExpand} and has later been integrated in CEGAR-style QBF solvers~\cite{JanotaKMC16/RAReQS,Tentrup/2017/OnExpansionAndResolution}. 

\begin{figure}[t]
\vspace{-2mm}
\begin{center}
\begin{tcolorbox}
\begin{minipage}{10cm}
$	
\phantom{Assume}~
\inferrule*[Left=Assume]{(\Ready,C,D,\chi,\alpha) \\ \var(l)\in D(0)}
          {(\Ready,C,D,\chi,\alpha\wedge l)}
$
\end{minipage}

\vspace{5mm}
\begin{minipage}{12.1cm}
$
\phantom{Close}~~
\inferrule*[Left=Close]{(\Ready,C,D,\chi,\alpha) \\ D=X\cup Y}
          {(\Ready,C(0),D(0),\chi \wedge \neg \alpha,\True)}
$
\end{minipage}
%
\end{tcolorbox}
\end{center}
\vspace{-5mm}
\caption{Inference rules adding case distinctions to ID}
\label{fig:casedistinctions}
\end{figure}

One way to look at the expansion of a universal variable $x$ is that it introduces a case distinction over the possible values of $x$ in the Skolem functions. 
However, instead of creating a copy of the formula explicitly, which often caused a blowup in required memory, we can reason about the two cases sequentially. 
The rules in Fig.~\ref{fig:casedistinctions} extend ID by universal expansion in this spirit.

Using \textsc{Assume} we can, at any point, assume that a variable $v$ in $D(0)$, i.e. a variable that has a unique Skolem function without any decisions, has a particular value.
This is represented by extending $\alpha$ by the corresponding literal of $v$, which restricts the domain of the Skolem function that we try to construct for subsequent $\deterministic$ and $\unconflicted$ checks.
Invariant~\ref{inv:Dconsistency} and Lemma~\ref{prop:soundnessTrue} already accommodate the case that $\alpha$ is not $\True$.

When we reach a point where $D$ contains all variables, we cannot apply \textsc{Sat}, as that requires $\alpha$ to be true. 
In this case, Invariant~\ref{inv:Dconsistency} only guarantees us that the function we constructed is correct on the domain $\chi\land\alpha$. 
We can hence restrict the domain for which we still need to find a Skolem function and strengthen $\chi$ by $\neg\alpha$. 
In particular, \textsc{Close} maintains Invariant~\ref{inv:psi}.
When $\chi$ ends up being equivalent to $\False$, Invariant~\ref{inv:psi} guarantees that the original formula is true. 
(In this case we can reach the $\SAT$ state easily, as we know that from now on every application of \textsc{Propagate} must succeed.
\footnote{Technically, we could replace \textsc{Sat} by a rule that allows us to enter the $\SAT$ state whenever $\chi$ is $\False$, which arguably would be more elegant. But that would require us to introduce the \textsc{Close} rule already for the basic ID inference system.})

Note that \textsc{Assume} does not restrict us to assumptions on single variables. 
Together with \textsc{Decide} and \textsc{Propagate} it is possible to introduce variables with arbitrary definitions, add them to $D(0)$, and then assume an outcome with the rule \textsc{Assume}.

\vspace{-2mm}
\paragraph{Example.} Again, we consider the formula from Subsection~\ref{ssec:example}.
Instead of the reasoning steps described in Subsection~\ref{ssec:example}, we start using \textsc{Assume} with literal $x_2$. 
Whenever checking $\deterministic$ or $\unconflicted$ in the following, we will thus restrict ourselves to universal assignments that set $x_2$ to true. 
It is easy to check that this allows us to propagate not only $y_1$ and $y_2$, but also $y_3$. 
A decision (e.g. $\delta'=\{(y_4)\}$) for $y_4$ allows us to also propagate $y_4$ (this time without potential for conflicts), arriving in state $(\Ready, \varphi.\delta', X\cup\{y_1,y_2,y_3,y_4\}, \True, x_2)$.

We can \textsc{Close} this case concluding that under the assumption $x_2$ we have found a Skolem function. 
We enter the state 
$(\Ready, \varphi, X, \neg x_2, \True)$ which indicates that in the future we only have to consider  universal assignments with $\neg x_2$. 
Also for the case $\neg x_2$, we cannot encounter conflicts for this formula. 
Expansion hence allows us to prove this formula without any conflicts.

\section{Experimental Evaluation}
\label{sec:eval}

We extended the QBF solver CADET~\cite{RabeSeshia/2016/IncrementalDeterminization} by the extensions described in Section~\ref{sec:CEGAR} and Section~\ref{sec:cases}. 
We use the CADET-IR and CADET-E to denote the extensions of CADET by inductive reasoning (Section~\ref{sec:CEGAR}) and universal expansion (Section~\ref{sec:cases}), respectively. 
We also combined both extensions and refer to this version as CADET-IR-E.
The experiments in this section evaluate these extensions against the basic version of CADET and against other successful QBF solvers of the recent years, in particular GhostQ~\cite{KlieberSGC/2010/GhostQ}, RAReQS~\cite{JanotaKMC16/RAReQS}, Qesto~\cite{JanotaMarquesSilva/2015/SolvingQBFByClauseSelection}, DepQBF~\cite{LonsingBiere/2010/DepQBF} in version 6, and CAQE~\cite{RabeTentrup/2015/CAQEACertifyingQBFSolver,Tentrup/2017/OnExpansionAndResolution}. 
For every solver except CADET and GhostQ, we use Bloqqer~\cite{BiereLS/2011/QBCE_Bloqqer} in version 031 as preprocessor. 
For our experiments, we used a machine with a $3.6\,\text{GHz}$ quad-core Intel Xeon processor and $32\,\text{GB}$ of memory.
The timeout and memout were set to $600$ seconds and $8\,\text{GB}$. 
We evaluated the solvers on the benchmark sets of the last competitive evaluation of QBF solvers, QBFEval-2017~\cite{conf/sat/Pulina16}.

\begin{figure}[t]
  \centering
  \begin{tikzpicture}
    \begin{semilogyaxis}[xlabel=\# solved instances,ylabel=time (sec.),width=.95\columnwidth,height=6cm,ymin=0.02,ymax=900,xmin=0,xmax=250,legend entries={GhostQ,CADET-IR-E,CADET-IR,RAReQS,CAQE,Qesto,CADET-E,DepQBF,CADET}, 
        transpose legend,
        legend columns=2,
        legend style={at={(0.45,-0.25)},anchor=north},
        ]
      \addplot+[dashed,very thick,mark repeat=10,mark=none,cyan] table {plots/qbfeval-2017-2qbf_cactus_ghostq-g0.dat};
      \addplot+[solid,very thick,mark=none,mark repeat=10,violet] table {plots/qbfeval-2017-2qbf_cactus_cadet-cegar-case-splits-g0.dat};
      \addplot+[solid,very thick,mark=none,mark repeat=10,blue] table {plots/qbfeval-2017-2qbf_cactus_cadet-cegar-g0.dat};
      \addplot+[dashed,thick,mark=none,mark repeat=10,black] table {plots/qbfeval-2017-2qbf_cactus_rareqs-g0.dat};
      \addplot+[dashed,thick,mark=none,mark repeat=10,orange] table {plots/qbfeval-2017-2qbf_cactus_caqe-g0.dat};
      \addplot+[dashed,thick,mark=none,mark repeat=10,brown] table {plots/qbfeval-2017-2qbf_cactus_qesto-g0.dat};
      \addplot+[solid,very thick,mark repeat=10,mark=none,red] table {plots/qbfeval-2017-2qbf_cactus_cadet-case-splits-g0.dat};
      \addplot+[dashed,very thick,mark repeat=10,mark=none,magenta] table {plots/qbfeval-2017-2qbf_cactus_depqbf-g0.dat};
      \addplot+[solid,very thick,mark repeat=10,mark=none,green] table {plots/qbfeval-2017-2qbf_cactus_cadet-g0.dat};
    \end{semilogyaxis}
  \end{tikzpicture}
  \label{fig:cactus_qbfeval2017}
  \vspace{-2mm}
  \caption{Cactus plot comparing solvers on the QBFEval-2017 2QBF benchmark.}
  \vspace{-1mm}
\end{figure}
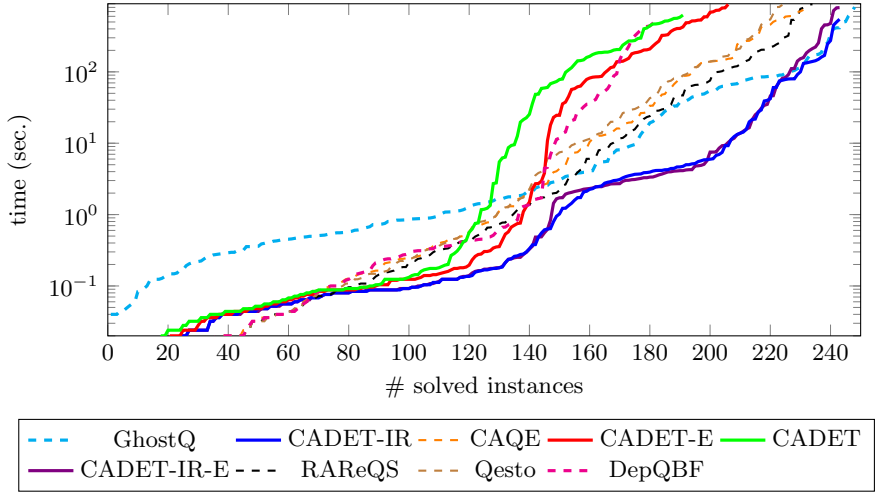

\vspace{-2mm}
\paragraph{How does inductive reasoning affect the performance? } 
In Fig.~\ref{fig:cactus_qbfeval2017} we see that CADET-IR clearly dominates plain CADET. 
It also dominates all solvers that relied on clause-level CEGAR and Bloqqer (CAQE, Qesto, RAReQS). 

Only GhostQ beats CADET-IR and solves 5 more formulas (of 384). 
A closer look revealed that there are many formulas for which CADET-IR and GhostQ show widely different runtimes hinting at potential for future improvement. 

GhostQ is based on the CEGAR principle, but reconstructs a circuit representation from the clauses instead of operating on the clauses directly~\cite{KlieberSGC/2010/GhostQ}. 
This makes GhostQ a representative of QBF solvers working with so called ``structured'' formulas (i.e. not CNF). 
CADET, on the other hand, refrains from identifying logic gates in CNF formulas and directly operates with the ``unstructured'' CNF representation. 
In the ongoing debate in the QBF community on the best representation of formulas for solving quantified formulas, our experimental findings can thus be interpreted as a tie between the two philosophies.

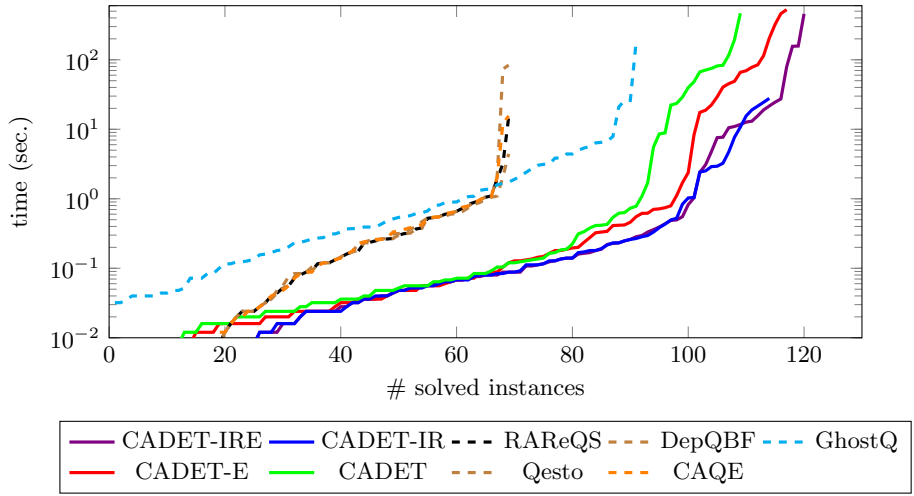
\begin{figure}[t]
  \centering
  \begin{tikzpicture}
    \begin{semilogyaxis}[xlabel=\# solved instances,ylabel=time (sec.),width=.95\columnwidth,height=6cm,ymin=0.01,ymax=600,xmin=0,xmax=130,legend entries={CADET-IRE,CADET-E,CADET-IR,CADET,RAReQS,Qesto,DepQBF,CAQE,GhostQ}, 
        transpose legend,
        legend columns=2,
        legend style={at={(0.5,-0.25)},anchor=north},
        ]
%

    \addplot+[solid,very thick,mark=none,violet] table {plots/hardwarefixpoint_cactus_cadet-cegar-case-splits-g0.dat};
      \addplot+[solid,very thick,mark=none,red] table {plots/hardwarefixpoint_cactus_cadet-case-splits-g0.dat};
      \addplot+[solid,very thick,mark=none,blue] table {plots/hardwarefixpoint_cactus_cadet-cegar-g0.dat};
      \addplot+[solid,very thick,mark=none,green] table {plots/hardwarefixpoint_cactus_cadet-g0.dat};
      \addplot+[dashed,very thick,mark=none,black] table {plots/hardwarefixpoint_cactus_rareqs-g0.dat};
      \addplot+[dashed,very thick,mark=none,brown] table {plots/hardwarefixpoint_cactus_qesto-g0.dat};
      \addplot+[dashed,very thick,mark=none,brown] table {plots/hardwarefixpoint_cactus_depqbf-5-g0.dat};
      \addplot+[dashed,very thick,mark=none,orange] table {plots/hardwarefixpoint_cactus_caqe-g0.dat};
    \addplot+[dashed,very thick,mark=none,cyan] table {plots/hardwarefixpoint_cactus_ghostq-g0.dat};
    \end{semilogyaxis}
  \end{tikzpicture}
  \vspace{-3mm}
  \caption{Cactus plot comparing solver performance on the Hardware Fixpoint formulas. Some but not all of these formulas are part of QBFEval-2017. The formulas encode diameter problems that are known to be hard for classical QBF search algorithms~\cite{DaijueYinleiDarshSharad/2005/DiameterProblems}.}
  \label{fig:cactus_cd}
\end{figure}

\vspace{-1mm}
\paragraph{Is the inductive reasoning extension just a portfolio-approach?}
To settle this question, we created a version of CADET-IR, called IR-only, that exclusively applies inductive reasoning by generating assignments to the universals and applying \textsc{InductiveReasoning}. 
This version of CADET does not learn any clauses, but otherwise uses the same code as CADET-IR. 
On the QBFEval-2017 benchmark, IR-only and CADET together solved 235 problems within the time limit, while CADET-IR solved 243 problems. 
That is, even though the combined runtime of CADET and IR-only was twice the runtime of CADET-IR, they solved fewer problems. 
CADET-IR also uniquely solved 22 problems. 
This indicates that CADET-IR improves over the portfolio approach. 

%
%
%
%

\vspace{-1mm}
\paragraph{How does universal expansion affect the performance? } 
CADET-E clearly dominates plain CADET on QBFEval-2017, but compared to CADET-IR and some of the other QBF solvers, CADET-E shows mediocre performance overall. 
However, for some subsets of formulas, such as the Hardware Fixpoint formulas shown in Fig.~\ref{fig:cactus_cd}, CADET-E dominated CADET, CADET-IR, and all other solvers. 
We also combined the two extensions of CADET to obtain \mbox{CADET-IR-E}. 
While this helped to improve the performance on the Hardware Fixpoint formulas even further, it did not change the overall picture on QBFEval-2017. 


\section{Conclusion}

Reasoning in quantified logics is one of the major challenges in com\-puter-aided verification. 
Incremental Determinization (ID) introduced a new algorithmic principle for reasoning in 2QBF and delivered first promising results~\cite{RabeSeshia/2016/IncrementalDeterminization}. 
In this work, we formalized and generalized ID to improve the understanding of the algorithm and to enable future research on the topic. 
The presentation of the algorithm as a set of inference rules has allowed us to disentangle the design choices from the principles of the algorithm (Section~\ref{sec:inference-rules}). 
Additionally, we have explored two extensions of ID that both significantly improve the performance:
The first one integrates the popular CEGAR-style algorithms and Incremental Determinization (Section~\ref{sec:CEGAR}). 
The second extension integrates a different type of reasoning termed universal expansion (Section~\ref{sec:cases}). 

\paragraph{Acknowledgements}
We want to thank Martina Seidl, who brought up the idea to formalize ID as inference rules, and Vijay D'Silva, who helped with disentangling the different perspectives on the algorithm. 
This work was supported in part by NSF grants 1139138, 1528108, 1739816, SRC contract 2638.001, the Intel ADEPT center, and the European Research Council (ERC) Grant OSARES (No. 683300). 
\phantom{\cite{SiekmannWrightson/1983/AutomationOfReasoning2ClassicalPapersOnComputationalLogic}}




\bibliography{main}

\end{document}